\documentclass[preprint,12pt,authoryear]{elsarticle}
\usepackage{amssymb}
\usepackage{amsthm}
\usepackage{setspace}
\usepackage{epsfig}
\usepackage{color}
\usepackage{bbm}
\usepackage{amssymb}
\usepackage{amsmath}
\usepackage{amsfonts}
\usepackage{natbib}
\usepackage{verbatim}
\usepackage{geometry}
\usepackage{graphicx}
\usepackage{float}
\usepackage{subfigure}
\newtheorem{theorem}{Theorem}[section]

\newtheorem{exa}{Example}[section]
\newtheorem{lem}{Lemma}[section]
\newtheorem{proposition}{Proposition}[section]
\newtheorem{corollary}{Corollary}[section]
\newtheorem{remark}{Remark}[section]
\newtheorem{definition}{Definition}[section]

\newtheorem{lemma}{Lemma}[section]

\usepackage{geometry}
\def\be{\begin{equation}} 
\def\ee{\end{equation}} 
\def\beqn{\begin{eqnarray}} 
\def\eeqn{\end{eqnarray}} 
\def\beq{\begin{eqnarray*}} 
\def\eeq{\end{eqnarray*}} 
\def\ba{\begin{array}} 
\def\ea{\end{array}} 
\newcommand{\bt}{\begin{theo}}
\newcommand{\et}{\end{theo}}
\newcommand{\bl}{\begin{lem}}
\newcommand{\el}{\end{lem}}
\newcommand{\bc}{\begin{co}}
\newcommand{\ec}{\end{co}}
\newcommand{\bp}{\begin{pr}}
\newcommand{\ep}{\end{pr}}
\newcommand{\bex}{\begin{exa}}
\newcommand{\eex}{\end{exa}\vspace{-4mm}}
\newcommand{\br}{\begin{re}}
\newcommand{\er}{\end{re}\vspace{-3mm}}

\journal{arXiv}
\begin{document}

\begin{frontmatter}

\title{Robust optimized certainty equivalents and quantiles for loss positions with distribution uncertainty
}
\tnotetext[mytitlenote]{The authors appreciated  the support of  the NSFC grant (No. 12671583).}

\author[math]{Zixin Feng}
\ead{zixin.feng@cumt.edu.cn}

\author[math]{Weiwei Li}
\ead{liweiwei\_vv@163.com}

\author[math]{Dejian Tian\corref{correspondingauthor}}
\ead{djtian@cumt.edu.cn}

\address[math]{School of Mathematics, China University of Mining and Technology, Xuzhou, P.R. China}
\cortext[correspondingauthor]{Corresponding author}

\begin{abstract}
This paper investigates the robust optimized certainty equivalents and analyzes their properties as risk measures under distribution uncertainty. Building on this, robust generalized quantiles are proposed and discussed. We then consider robust expectiles with two specific penalization functions. For the one with a linear penalization function, it is proved to be a coherent risk measure and its dual representation is provided. Furthermore, numerical simulations are conducted to examine the effect of the penalization functions on the robust expectiles and to compare them with classical expectiles.
\end{abstract}

\begin{keyword}
robust optimized certainty equivalents \sep robust quantiles \sep robust expectiles \sep distribution uncertainty

\end{keyword}

\end{frontmatter}


\section{Introduction}

From a risk regulator's point of view, \cite{ADEH97} and \cite{ADEH99} first propose the axiomatic definition of coherent risk measures, which satisfy the following properties: monotonicity, translation invariance, positive homogeneity and subadditivity. Then, \cite{H00}, \cite{FR02}, \cite{FS02} independently introduce convex risk measures, which also need to satisfy convexity.

Based on the degree of risk aversion of financial agents,  \cite{BT86} put forward the \textit{optimized certainty equivalent}, which is a decision theoretic criterion on some utility functions. The optimized certainty equivalents consider the extremum problem of a stochastic nonlinear programming.  \cite{BT07} further examine that the optimized certainty equivalents satisfy the axiomatic definition of convex risk measures and obtain their dual representation theorems.

Risk measures can also be defined by the optimal solution of the extremum problem such as the VaR (Value at Risk).   VaR, a quantile function of the loss position,  is a simple and reasonable risk measure established by risk regulators for the banking industry; see \cite{DP97}.  VaR has good properties, such as homogeneity, translation invariance and monotonicity. However, as a risk measure, VaR is not coherent or convex, and it cannot capture the tail risk and does not pay attention to the scale of loss. Therefore, scholars have attempted to construct convex risk measures or coherent risk measures based on VaR from the following two perspectives.

The first perspective is based on the VaR itself. For example, \cite{ADEH99} and \cite{D02} propose WCE (Worst conditional expectation) and TCE (Tail conditional expectation). Further, \cite{A02}, \cite{AT02}, \cite{RU02}, \cite{T02}, \cite{C06} and other scholars introduce Conditional VaR, Expected shortfall, Tail VaR, Average VaR, Weight VaR, etc.  Thus, a more reasonable and effective risk measurement model system is gradually established. The reader can refer to Chapter 4 of \cite{FS16}.

The other perspective is based on the equivalent characterization of VaR. \cite{KB78} propose that VaR$_\alpha(X)$ is  the solution of the following minimization problem
\begin{align*}
\min _{m \in \mathbb{R}}\left\{\alpha E_P\left[(X-m)^{+}\right]+(1-\alpha) E_P\left[(X-m)^{-}\right]\right\}, \alpha \in (0,1).
 \end{align*}
Based on this minimization problem, several quantiles have been introduced considering more general loss functions.  \cite{NP87} introduce the expectiles as the minimizers of a piecewise quadratic loss function.  \cite{BC88} discuss M-quantiles and \cite{C96} considers the power loss function, respectively.  \cite{BKMR14} study the generalized quantiles, which is the optimal solution of
\begin{align*}
 \underset{m \in \mathbb{R}}{\operatorname{min}}\left\{\alpha E_P\left[l_{1}\left((X-m)^{+}\right)\right]+(1-\alpha) E_P\left[l_{2}\left((X-m)^{-}\right)\right]\right\},
\end{align*}
where $l_{1}$, $l_{2}$ are increasing convex functions, and the only generalized quantiles as coherent risk measures are the expectiles with $\alpha>1/2$. Besides, \cite{M18} propose new generalized quantiles based on rank-dependent expected utility (RDEU).  Generalized quantiles have been shown to have important applications in risk measurement and mathematical finance.  For example, the reader can refer to \cite{TD02}, \cite{M15}, \cite{C19} and \cite{XZH23} etc.


Whether the optimized certainty equivalents or the generalized quantiles, we assume that we know the exact distribution of the loss position in a given probability space. However, the distributions of future losses are uncertain. Recently, \cite{BDT20} consider the distribution uncertainty problems and
propose the robust optimized certainty equivalent, which is defined as\footnote{In fact, given a prior distribution $\mu_{0}$, the original definition of the robust optimized certainty equivalents by \cite{BDT20} is
\begin{align*}
\mathcal{OCE}(l):=\inf _{m \in \mathbb{R}}\left\{m+\sup _{\mu \in \mathcal{M}(\mathbb{R})}\left(\int_{\mathbb{R}} l(x-m)\mu(dx)-\varphi\left(d_{c}\left(\mu_{0}, \mu\right)\right)\right)\right\}.
\end{align*}In our situation, we emphasize the loss position $X$. }
\begin{align}\label{ocex}
\mathcal{OCE}_{l, \varphi}(X):=\inf _{m \in \mathbb{R}}\left\{m+\sup _{\mu \in \mathcal{M}(\mathbb{R})}\left(\int_{\mathbb{R}} l(x-m)\mu(dx)-\varphi\left(d_{c}\left(\mu_{X}, \mu\right)\right)\right)\right\},
\end{align}
with $l:\mathbb{R}\to\mathbb{R}$ the loss function, $\mathcal{M}(\mathbb{R})$ the set of probability on $\mathbb{R}$,   $\varphi$ being a penalization function, $d_{c}$ the Wasserstein distance with cost function $c(x,y)$, $\mu_{X}$ a given  distribution of loss position $X$.  We generally consider $\mu_{X}$ is likely to the true distribution of future losses.

Motivated by the robust optimized certainty equivalents, for any two functions $l_{1}$ and $l_{2}$, and $h(x)=\alpha l_{1}\left(x^{+}\right)+(1-\alpha)l_{2}\left(x^{-}\right)$ with $\alpha \in (0,1)$, this paper proposes the \textit{robust generalized quantiles} $q_{\alpha,\varphi}(X)$ for $X$ which satisfy
$$q_{\alpha,\varphi}(X) \in \underset{m \in \mathbb{R}}{\operatorname{argmin}}\hspace{0.3em} \left\{\sup _{\mu \in \mathcal{M}(\mathbb{R})}\left(\int_{\mathbb{R}} h(x-m)\mu(dx)-\varphi\left(d_{c}\left(\mu_{X}, \mu\right)\right)\right)\right\}.$$ The robust generalized quantiles are the natural generalization of quantiles for loss positions with distribution uncertainty.

The paper contributes to the literature in the following three aspects.  First, in order to investigate the loss positions with distribution uncertainty,  the definition of the robust optimized certainty equivalents is refined.  On this basis,  we can study their properties as risk measures and find that they satisfy translation invariance, monotonicity and convexity (see Proposition \ref{prop:2-1}). Furthermore, Proposition \ref{prop:2-2} analyzes the influences of the loss function and the penalization function on the robust optimized certainty equivalents. The loss function reflects the agent's risk aversion,  while the penalization function reflects the agent's robust aversion. The larger loss function implies the stronger risk aversion,  and the larger penalization function implies the weaker robust aversion, which means the higher the reliability of the priori distribution. A related reachability condition is also given for the robust optimized certainty equivalents in Proposition \ref{prop:2-3}.

Second, we propose robust generalized quantiles which incorporate with the distribution uncertainty of the loss positions.  Proposition \ref{prop:2-4} provides the sufficient conditions for the existence of robust generalized quantiles.  For any
$\alpha\in(0,1)$, when $h(x)=\alpha x^++(1-\alpha) x^-$ and  the cost function $c(x,y)=|x-y|$, Proposition \ref{prop:2-5} shows that robust generalized quantiles degenerates into classical VaR for any penalization function $\varphi$. It indicates that the robust generalized quantiles display the robustness for penalization function in this specification.

Third,  we consider two kinds of robust generalized quantiles by introducing two specific penalization functions $\varphi_1(x)=\delta_1 x$ with $\delta_1>max\{\alpha,1-\alpha\}$ and $\varphi_2(x)=\infty I_{(\delta_2,+\infty)}(x)$ with  $\delta_{2}>0$.  We call them the \textit{robust expectiles with $\varphi_1$} and the \textit{robust expectiles with $\varphi_2$} respectively.  Using a dual formula,  we transform the solving for the robust expectiles problem into the minimum problem of finite dimensions,  so as to further study their properties as risk measures and the impact of penalization parameters on them.  We find that the robust expectiles with $\varphi_1$  are coherent risk measures when $\alpha >1/2$,  and we establish the dual representation theorems (see Theorem \ref{th:3-1} and Theorem \ref{th:3-2}). Robust expectiles with $\varphi_2$ are also studied.  In addition, we also provide comparisons between the robust expectiles and the expectiles under some specific prior distributions.

The paper is organized as follows. Section \ref{sec2} considers the properties of robust optimized certainty equivalents  and proposes the definition of the robust generalized quantiles. Section \ref{sec3} focuses on two kinds of specific robust expectiles corresponding to two popular penalization functions. All the proofs are relegated to Section \ref{sec4}. Section \ref{sec5} concludes the paper.

\section{Robust optimized certainty equivalents and generalized quantiles}\label{sec2}

In this section, after introducing the robust optimized certainty equivalents, originally proposed by \cite{BDT20}, we investigate its properties and give the definition of the robust generalized quantiles.

\subsection{Robust optimized certainty equivalents}

Let $(\Omega, \mathcal{F})$ be a measurable space, and there exists a priori probability measure $P$ on it.  Let $X$ be a random variable from $(\Omega, \mathcal{F})$ to $(\mathbb{R}, \mathcal{B}{(\mathbb{R})})$. The  priori  distribution  or law of $X$ is defined as the probability measure on the line given by
$$\mu_X(A):=(P\circ X^{-1})(A):=P(X^{-1}(A)), \qquad \forall A\in\mathcal{B}{(\mathbb{R})}. $$
Then, $(\mathbb{R}, \mathcal{B}{(\mathbb{R})}, \mu_X)$ is a probability space, and $\mu_X$ is  the priori  distribution of $X$. Let $\mathcal{M}(\mathbb{R})$ denote the set of all probability measures defined on the Borel $\sigma$-algebra $\mathcal{B}{(\mathbb{R})}$.
Then, for any measurable and bounded from below function $l:\mathbb{R}\to \mathbb{R}$,  we can define the robust optimized certainty equivalents with respect to $X$, i.e.,
\begin{align}\label{eq:21}
\mathcal{OCE}_{l, \varphi}(X):=\inf _{m \in \mathbb{R}}\left\{m+\mathcal{E}_{\varphi}(l,X,m)\right\},
\end{align}
where the nonlinear functional $\mathcal{E}_{\varphi}(l, X,m)$ is defined as
\begin{equation}\label{eq:2-2}
\mathcal{E}_{\varphi}(l,X,m):=\sup _{\mu \in \mathcal{M}(\mathbb{R})}\left(\int_{\mathbb{R}} l(x-m)\mu(dx)-\varphi\left(d_{c}\left(\mu_{X}, \mu\right)\right)\right),
\end{equation}
where $\varphi$ is a penalization function, $d_{c}$ is a distance with cost function $c(\cdot,\cdot)$ such as the Wasserstein distance.

We give some specifications as follows:
\begin{itemize}
\item A loss function $l$:  $\mathbb{R}\to \mathbb{R}$, which is measurable and bounded from below.
\item A penalization function $\varphi$: $[0,+\infty]\rightarrow[0,+\infty]$, which is convex, increasing, lower semicontinuous with $\varphi(0)=0$.  $\varphi^{*}$ is the convex conjugate of $\varphi$, with $\varphi^{*}(y)=\sup _{x\geq0}(xy-\varphi(x))$.  $\varphi$ and $\varphi^{*}$ are not constants.
\item The cost function $c(x,y)=|x-y|^p$ with $p\geq 1$, for all $x,y\in \mathbb{R}$.
\item The distance $d_{c}$ between $\mu_{X}$ and $\mu$  : for any $\mu \in \mathcal{M}(\mathbb{R})$,
$$
d_{c}\left(\mu_{X}, \mu\right):=\inf \left\{\int_{\mathbb{R} \times \mathbb{R}} c(x, y) \pi(d x, d y): \begin{array}{l}
\pi \in \mathcal{M}(\mathbb{R} \times \mathbb{R}) \text { such that } \\
\pi(\cdot \times \mathbb{R})=\mu_{X} \text { and } \pi(\mathbb{R} \times \cdot)=\mu.
\end{array}\right\}.
$$
\end{itemize}

Unlike \cite{BDT20},  we emphasize the priori distribution of random variable $X$ in the definition of robust optimized certainty equivalents.  Here, we treat $\mu_X$ as the fixed baseline distribution of the random loss $X$.  The reason why the distance $d_{c}$ is a popular choice to model the ambiguity distribution is that one has $d_{c}(\mu_{X},\mu_{n}) \rightarrow 0$ if and only if $\mu_{n}$ converges weakly to $\mu_{X}$ and $\int_{\mathbb{R}} |x|^p \mu_n(dx) \to \int_{\mathbb{R}} |x|^p \mu_X(dx)$ for $p\geq 1$, see \cite{V08}. This means that we can use $\varphi(d_{c}(\cdot,\cdot))$ to penalize those distributions that are far away from the baseline distribution $\mu_{X}$ accurately.  Hence, the definition of robust optimized certainty equivalent can be used to describe the distribution uncertainty of the random variables.

Fixing a prior distribution $\mu_{0}$,  \cite{BDT20} mainly solve the computational problem. They does not  stress the loss position $X$ and not take into account the properties of $\mathcal{OCE}_{l, \varphi}(\cdot)$ and not address the effects for different random variables.   Motivated by \cite{BT86} and \cite{BT07},  we consider some properties for robust optimized certainty equivalents in this paper.

\begin{proposition}\label{prop:2-1}
Let a loss function $l:\mathbb{R}\to \mathbb{R}$ be convex and increasing and $\varphi$ be a penalization function.  Then the following properties hold.
\begin{itemize}
\item[(a)] Prior distribution invariance: If $X$ and $Y$ have the same prior distribution under $P$, then   $\mathcal{OCE}_{l, \varphi}(X) =\mathcal{OCE}_{l, \varphi}(Y)$.
\item[(b)] Translation invariance: $\mathcal{OCE}_{l, \varphi}(X+C)=\mathcal{OCE}_{l, \varphi}(X)+C$, for any $C\in \mathbb{R}$.
\item[(c)] Monotonicity: If $X\leq Y$, $P$-a.s.,  then $\mathcal{OCE}_{l, \varphi}(X) \leq \mathcal{OCE}_{l, \varphi}(Y)$.
\item[(d)]  Convexity: For any random variables $X$ and $Y$,  and any $t\in(0,1)$, one has
$$\mathcal{OCE}_{l, \varphi}(tX+(1-t)Y)\leq t\mathcal{OCE}_{l, \varphi}(X) + (1-t)\mathcal{OCE}_{l, \varphi}(Y).$$
\end{itemize}
\end{proposition}

\begin{remark}
Although the robust optimized certainty equivalent $\mathcal{OCE}_{l, \varphi}(\cdot)$
involves the distribution uncertainty of random variables,   $\mathcal{OCE}_{l, \varphi}(\cdot)$ still may satisfy some good properties such as monotonicity, translation invariance and convexity.   It is worth noting that $\mathcal{OCE}_{l, \varphi}(\cdot)$ does not necessarily satisfy the property of preserving constants.  For example,  taking $p=1$, $l(x)=1+x^{+}$, for any penalization function $\varphi(\cdot)$,  it is easy to verify that $\mathcal{OCE}_{l, \varphi}(0)=1$.
\end{remark}
\begin{remark}\label{rem-2.2}
An another natural way  to consider uncertainty is not in the realm of image measures, but rather on the probability space itself.  To be specific, one can define
\begin{align}\label{eq:bdow21}
\mathcal{E}_{\varphi}(l,X,m):=\sup _{\mu \in \mathcal{M}(\mathbb{R}^n)}\left(\int_{\mathbb{R}^n} l(\psi_X(s)-m)\mu(ds)-\varphi\left(d_c(\mu_S, \mu)\right)\right),\end{align}
where vector $S$ is describing all randomness in a market, then every option $X$ on $S$ is of the form $X =\psi(S)$ for a certain function  $\psi=\psi_X$.  Based on this, one can also define the resulting robust optimized certainty equivalent. This choice is the more natural generalization that stresses the dependence on the loss. However, this formulation may be
less likely to obtain explicit formulas for the following generalized quantiles.
\end{remark}

\begin{proposition}\label{prop:2-2} Let $X$ be a random variable on the prior probability space $(\Omega, \mathcal{F},P)$. Given some loss functions $f$,   $g$ and $l$,  and some  penalization functions $\phi$, $\psi$ and $\varphi$.  Then the following properties hold.
\begin{itemize}
\item[(a)] If  $f\geq g$, then $\mathcal{OCE}_{f, \varphi}(X) \geq\mathcal{OCE}_{g, \varphi}(X)$.
\item[(b)]  If  $\phi\geq \psi$, then $\mathcal{OCE}_{l, \psi}(X)\geq\mathcal{OCE}_{l, \phi}(X)\geq OCE_{l}(X)$, where
$$OCE_{l}(X)=\inf _{m \in \mathbb{R}}\left\{m+
E_{P}[l(X-m)]\right\}.$$

\end{itemize}
\end{proposition}


Proposition \ref{prop:2-2} establishes a relationship between the loss function, penalty function, and the optimal certainty equivalent.  Specifically, higher values of the loss function induce an increased $\mathcal{OCE}$, while stronger penalization corresponds to a reduced  $\mathcal{OCE}$. This relation reflects the epistemic interpretation of the penalty function--it quantifies the decision agent's confidence in the baseline distribution. Consequently, enhanced penalization drives it to converge towards the theoretical bound ${OCE}_{l}(X)$  derived from the prior distribution $\mu_{X}$, implying that greater penalty strength indicates higher trustworthiness of the prior distributional assumption.

\subsection{Robust generalized quantiles}
As what we have expected, robust optimized certainty equivalents based on the nonlinear functional $\mathcal{E}_{\varphi}(l, X,m)$ have good properties as risk measures.  Motivated by \cite{BKMR14},  we consider the generalized quantiles  under robust distributions with $\mathcal{E}_{\varphi}(l, X,m)$.

Let $l_1, l_2:[0,+\infty)\rightarrow [0,+\infty)$ be two convex and increasing  loss functions.  For any $\alpha \in (0,1)$,
\begin{align}\label{eq:2.2-1}
h(x):=\alpha l_1(x^+)+(1-\alpha) l_2(x^-), ~~\forall x\in \mathbb{R}.
\end{align}
Then $h$ is a convex loss function.  For a random variable $X$ and a penalization function $\varphi$,   now we consider the following minimization problem
\begin{align*}
\pi_{\alpha}(X):=\underset{m \in \mathbb{R}}{\operatorname{inf}}~\mathcal{E}_{\varphi}(h,X,m),
\end{align*}where $\mathcal{E}_{\varphi}(h,X,m)$ is defined by \eqref{eq:2-2}.
And we call $q_{\alpha,\varphi}(X)$ the \textit{robust generalized quantiles} of $X$ if one has
$$q_{\alpha,\varphi}(X) \in \underset{m \in \mathbb{R}}{\operatorname{argmin}}\hspace{0.3em}\mathcal{E}_{\varphi}(h,X,m).$$

Now, we give a sufficient condition for the existence of robust generalized quantiles.
\begin{proposition}\label{prop:2-4}
Let $l_1, l_2:[0,+\infty)\rightarrow [0,+\infty)$ be two convex and increasing loss functions. For each $\alpha \in (0,1)$,  a random variable $X$ and a penalization function $\varphi$,  then it follows that
\begin{itemize}
\item[(a)] $\mathcal{E}_{\varphi}(h,X,m)$ is convex with respect to $m$, and
$$
\lim _{m \rightarrow-\infty} \mathcal{E}_{\varphi}(h,X,m)=\lim _{m \rightarrow+\infty} \mathcal{E}_{\varphi}(h,X,m)=+\infty .
$$
\item[(b)] Suppose $\mathcal{E}_{\varphi}(h,X,0)<+\infty $,  then there exists a closed interval $[m_1,m_2]$, such that
$$[m_1,m_2]=\underset{m \in \mathbb{R}}{\operatorname{argmin}}\hspace{0.3em} \mathcal{E}_{\varphi}(h,X,m).$$
\end{itemize}
\end{proposition}
Compared with generalized quantiles investigated by \cite{BKMR14}, robust generalized quantiles consider the uncertainty distributions of future losses, which also leads to an infinite dimensional problem of calculation. Fortunately, we can use the dual formula obtained by \cite{BDT20} to transform it into a finite dimensional convex function to solve the extremum problem.  More specifically,  since  $h(\cdot)$, defined in \eqref{eq:2.2-1}, is a loss function,  and its $\lambda c$-transform can be written as follows
$$
h^{\lambda c}(x)=\sup \{h(y)-\lambda |x-y|^p: y \in \mathbb{R} \text { such that } h(y)<\infty\}, ~~x\in\mathbb{R}.
$$
Using the dual formula (Lemma \ref{lem:4-1}),  we can obtain
\begin{align}\label{eq:2.2-2}
\mathcal{E}_{\varphi}(h,X,m)=\inf _{\lambda \geq 0}\left\{E_{P}[h^{\lambda c}(X-m)]+\varphi^*(\lambda)\right\}.
\end{align}

Proposition \ref{prop:2-4}(b) obtains a sufficient condition for the well posedness of the robust generalized quantiles. The following lemma provides another sufficient condition, avoiding  $h^{\lambda c}(\cdot)\equiv +\infty$, by controlling the growth rate of the loss function. This lemma also shows what kind of loss functions should be studied in the next section. 

\begin{lemma}\label{lem:2-1}
For a loss function $h:\mathbb{R}\to \mathbb{R}$, suppose that there exists a constant $C\geq 0$ such that for all $x\in \mathbb{R}$, $h(x)\leq C(1+|x|^p)$, where $p\geq 1$ is the power order for the cost function $c(\cdot,\cdot)$.  Then, there exists a constant $\lambda^* > C$, such that  $h^{\lambda^* c}(x)<+\infty$ for all $x \in\mathbb{R}$.
\end{lemma}

In particular, when $l_1(x)=l_2(x)=x$, then the corresponding robust generalized quantiles is called  \textit{robust VaR}.  Based on Lemma \ref{lem:2-1},  we should choose the cost function with $p\geq 1$.  We find that the robust VaR degenerates into VaR for any penalization function $\varphi$,  when the cost function $c(x,y)=|x-y|$,  which means that  the VaR itself has robustness.
\begin{proposition}\label{prop:2-5}
Suppose the cost function is $c(x,y)=|x-y|$, the loss functions $l_1(x)=l_2(x)=x$ , $h$ is defined by \eqref{eq:2.2-1} for each $\alpha\in(0,1)$,   and $X\in L^{1}(\Omega, \mathcal{F}, P)$.   Then, for any  penalization function $\varphi$ and for any $\alpha\in(0,1)$,   the robust VaR  can be degenerated into the classical VaR,  i.e.,  $$q_{\alpha,\varphi}(X)=VaR_{\alpha}(X).$$
\end{proposition}

\section{Robust expectiles }\label{sec3}

Although Lemma \ref{lem:2-1} tells us that the order of the penalty function can exceed that of the loss function, we find that it is only when their orders are equal that it brings convenience to the calculation. For VaR situation, the loss function is of the first order, so we take the cost function $c(x,y)=|x-y|$ as in Proposition \ref{prop:2-5}. 

\cite{NP87} have considered the expectiles for random variables.  \cite{BKMR14} establish the expectiles with the relationship for the risk measures. This section will investigate robust expectiles under distribution uncertainty.

Suppose $X\in L^{2}(\Omega, \mathcal{F},P)$. Choosing $l_1(x)=l_2(x)=x^2$, for each
 $\alpha \in (0,1)$, then
\begin{align}\label{eq:3-1}
h(x)=\alpha (x^+)^{2}+(1-\alpha) (x^-)^{2}, ~~\forall x\in \mathbb{R}.
\end{align}
$h$ is a convex loss function and we will choose the cost function $c(x,y)=|x-y|^2$.  

We will propose two kinds of robust generalized expectiles for random variables with uncertainty distributions by introducing the following two specific penalization functions $\varphi_{1}$ and $\varphi_{2}$:
\begin{itemize}
\item  $\varphi_1(x)=\delta_1 x$ with $\delta_1>max\{\alpha,1-\alpha\}$;
\item  $\varphi_2(x)=\infty I_{(\delta_2,+\infty)}(x)$ with  $\delta_{2}>0$.
\end{itemize}
By using dual formula,  we transform the problem of solving for the robust generalized quantiles into the minimum problem of finite dimensions, so as to further study their properties as risk measures and the impact of penalization parameters on them.

\subsection{Robust expectiles with $\varphi_1$}
This subsection considers the first specific penalization function $\varphi_1(x)=\delta_1 x$,  $x\geq 0$,  with $\delta_1>max\{\alpha,1-\alpha\}$, which grows linearly with respect to the distance, and it is named by the robust expectiles with $\varphi_1$.

\begin{definition}
Suppose $X\in L^{2}(\Omega, \mathcal{F},P)$.  For any $\alpha\in(0,1)$, let $h$ be the loss function defined by \eqref{eq:3-1}. The cost function $c(x,y)=|x-y|^2$. Then the robust expectiles with penalization function $\varphi_1$ are defined by:
\begin{align*}
e_{\alpha,\varphi_1}(X):=\underset{m \in \mathbb{R}}{\operatorname{argmin}}~\mathcal{E}_{\varphi_1}(h,X,m).
\end{align*}
\end{definition}

\begin{remark}
Lemma \ref{lem:2-1} explains why we choose $c(x,y)=|x-y|^2$ when we define the robust expectiles with $\varphi_1$. Using the dual formula, we can easily obtain
if the cost function $c(x,y)=|x-y|$, and we find that $h^{\lambda c}(x)\equiv +\infty$, which leads to $\mathcal{E}_{\varphi_1}(h,X,m)\equiv +\infty$ with $m\in \mathbb{R}$. It is a meaningless question.  Therefore, for the loss function with $h$, we should choose the cost function with $p\geq2$.
\end{remark}

The following Proposition \ref{prop:3-1} gives a representation of the robust expectiles with $\varphi_1$. Compared to its definition, this characterization is very straightforward. By means of the duality theorem,  the influence of its uncertain distribution is described by the penalty parameter $\delta_{1}$ and its prior distribution under $P$.

\begin{proposition}\label{prop:3-1}
Let $\alpha\in (0,1)$, and $e_{\alpha,\varphi_1}(X)$ be the robust expectiles with $\varphi_1$ of $X$, Then
\begin{align*}
e_{\alpha,\varphi_1}(X)=\underset{m \in \mathbb{R}}{\operatorname{argmin}} ~g_1(X,m,\delta_{1},\alpha),
\end{align*}
where
\begin{align}\label{eq:function-g}
g_1(X,m,\delta_{1},\alpha)=\frac{\alpha\delta_1}{\delta_1-\alpha}E_P\left[((X-m)^{+})^2\right]+\frac{(1-\alpha)\delta_1}{\delta_1-(1-\alpha)}E_P\left[((X-m)^{-})^2\right].\end{align}
\end{proposition}


Similar to the expectiles in the classical situation, we can also establish the  relationship between robust expectiles with penalization function $\varphi_1$ and  risk measures.  The following two theorems prove the robust expectile with $\varphi_1$ is a coherent risk measure and give its representation theorem.

\begin{theorem}\label{th:3-1}
Suppose that $\alpha \in (\frac{1}{2},1)$.  Then,  the robust expectile  $e_{\alpha,\varphi_1}(\cdot)$ is a coherent risk measure on $L^{2}(\Omega,\mathcal{F},P)$.
\end{theorem}

In the proof procedure of Theorem \ref{th:3-1}, for any $X$ in $L^{2}(\Omega,\mathcal{F},P)$, $\alpha\in(0,1)$ and $\delta_1>max\{\alpha,1-\alpha\}$,
we know that  $g_1'(X,  e_{\alpha,\varphi_1}(X), \delta_{1},\alpha)=0$.
Let  $\psi:\mathbb{R} \to \mathbb{R}$ be defined by
$$\psi(x):=\frac{2\alpha\delta_1}{\delta_1-\alpha} x^{+}-\frac{2(1-\alpha)\delta_1}{\delta_1-(1-\alpha)}x^{-},  ~~x\in\mathbb{R},$$
which is increasing and convex with respect to $x$ when $\alpha$ in $(\frac{1}{2},1)$.
Then, in this situation, $e_{\alpha,\varphi_1}(X)$ also satisfies
$$E_P\left[\psi(X- e_{\alpha,\varphi_1}(X))\right]=0.$$
Hence, when $\alpha \in (\frac{1}{2},1)$, the robust expectile can be regarded as a special case of shortfall risk measure which is given by
$$\rho(X)=\inf \{m\in \mathbb{R} \hspace{0.2em} | \hspace{0.2em} E_P[l(X-m)]\leq x_0 \},$$where $l$ is a loss function and $x_{0}$ is a ceiling for expected loss.  Then by the representation theorem for expected shortfall risk measure (see, \cite{FS02} or \cite{FS16}), we can obtain the  representation theorem for the robust expectile with $\varphi_1$.

\begin{theorem}\label{th:3-2}
For any $X\in L^{\infty}(\Omega,\mathcal{F},P)$, $\alpha\in(0,1)$ and $\delta_1>max\{\alpha,1-\alpha\}$,  the robust expectile $e_{\alpha,\varphi_1}(X)$ has the following dual representation
$$
e_{\alpha,\varphi_1}(X)= \begin{cases}
\underset{Q \in \mathcal{M}_1(P)}{\operatorname{max}} E_{Q}[X], & \text { if } \alpha \in (\frac{1}{2},1), \\
\underset{Q \in \mathcal{M}_2(P)}{\operatorname{min}} E_{Q}[X], & \text { if } \alpha \in (0,\frac{1}{2}), \\
\end{cases}
$$
where
$$
\mathcal{M}_1(P)=\left\{Q: \begin{array}{l}
Q\hspace{0.3em} is\hspace{0.3em} absolutely\hspace{0.3em} continuous\hspace{0.3em} with\hspace{0.3em} respect\hspace{0.3em} to\hspace{0.3em} P, \\
and \hspace{0.3em} \exists\hspace{0.2em} t_0 >0, \hspace{0.3em} such \hspace{0.3em} that \hspace{0.3em} \frac{2(1-\alpha)\delta_1}{\delta_1-(1-\alpha)}\leq t_0\frac{\mathrm{d}Q}{\mathrm{d}P} \leq \frac{2\alpha\delta_1}{\delta_1-\alpha}.
\end{array}\right\},
$$
and
$$
\mathcal{M}_2(P)=\left\{Q: \begin{array}{l}
Q\hspace{0.3em} is\hspace{0.3em} absolutely\hspace{0.3em} continuous\hspace{0.3em} with\hspace{0.3em} respect\hspace{0.3em} to\hspace{0.3em} P, \\
and \hspace{0.3em} \exists\hspace{0.2em} t_0 >0, \hspace{0.3em} such \hspace{0.3em} that \hspace{0.3em} \frac{2\alpha\delta_1}{\delta_1-\alpha}\leq t_0\frac{\mathrm{d}Q}{\mathrm{d}P} \leq \frac{2(1-\alpha)\delta_1}{\delta_1-(1-\alpha)}
\end{array}\right\}.
$$
\end{theorem}

Compared with Proposition 8 in \cite{BKMR14}, where they establish a  representation theorem for the generalized expectile when the random variable has no uncertainty distribution.   Our result in Theorem \ref{th:3-2} is not trivial, which provides a representation of the robust expectile for the random variable with uncertainty distributions.  We find that the existence of distribution uncertainty and the penalty function are significant for the proof process.  The penalization parameter $\delta_{1}$  appears explicitly in the set of probability measures $\mathcal{M}_{1}(P)$ or $\mathcal{M}_{2}(P)$.

\begin{corollary}
 For $X\in L^{2}(\Omega,\mathcal{F},P)$, $\alpha\in(0,1)$ and $\delta_1>max\{\alpha,1-\alpha\}$,  then  there exists a unique $m_1^*$ such that
$g_1'(X,m_{1}^{*}, \delta_{1},\alpha)=0,$ ~ i.e.,
\begin{align*}
-\frac{\alpha\delta_1}{\delta_1-\alpha}E_P\left[(X-m_{1}^{*})^{+}\right]+\frac{(1-\alpha)\delta_1}{\delta_1-(1-\alpha)}E_P\left[(X-m_{1}^{*})^{-}\right]=0,
\end{align*}where $e_{\alpha,\varphi_1}(X)=m_1^*$ and $g_1$ is defined in \eqref{eq:function-g}.    In particular, $m_1^* \equiv E_P[X]$ when $\alpha=\frac{1}{2}$.
\end{corollary}

In subsection 3.3, we will consider the impact of the penalization function on the robust expectiles with $\varphi_1$ under a specific baseline distribution.

\subsection{Robust expectiles with $\varphi_2$}

This subsection considers the second specific penalization function $\varphi_2(x)=\infty I_{(\delta_2,+\infty)}(x)$,  $x\geq 0$,  with $\delta_2>0$.
It is named by the robust expectiles with $\varphi_2$.   For any $\mu$ in $\mathcal{M}(\mathbb{R})$, $\varphi_{2}(d_{c}\left(\mu_{X}, \mu\right))=\infty$ if $d_{c}\left(\mu_{X}, \mu\right)>\delta_{2}$, which means the potential  uncertainty distributions $\mu$ of $X$ should satisfy $d_{c}\left(\mu_{X}, \mu\right)\leq \delta_{2}$.

\begin{definition}
Suppose $X\in L^{2}(\Omega, \mathcal{F},P)$.  For any $\alpha\in(0,1)$, let $h$ be the loss function defined by \eqref{eq:3-1}. The cost function $c(x,y)=|x-y|^2$.    Then the robust expectiles with penalization function $\varphi_2$ are defined by:
\begin{align*}
e_{\alpha,\varphi_2}(X):=\underset{m \in \mathbb{R}}{\operatorname{argmin}}~\mathcal{E}_{\varphi_2}(h,X,m).
\end{align*}
\end{definition}

\begin{proposition}\label{prop:3-2}
Let $\alpha \in (0,1)$, and $e_{\alpha,\varphi_2}(X)$ be the robust expectiles of $X$. Then
\begin{align*}
e_{\alpha,\varphi_2}(X)=\underset{m \in \mathbb{R}}{\operatorname{argmin}}\left\{\inf _{\lambda >max\{\alpha,1-\alpha\} }g_2(X,m,\lambda,\alpha) \right\},
\end{align*}
where
\begin{align}\label{eq:function-g2}
g_2(X,m,\lambda,\alpha)=\frac{\alpha\lambda}{\lambda-\alpha}E_P\left[((X-m)^{+})^2\right]+\frac{(1-\alpha)\lambda}{\lambda-(1-\alpha)}E_P\left[((X-m)^{-})^2\right]
+\delta_2\lambda.
\end{align}
\end{proposition}

\begin{corollary}\label{cor:3.2}
 Let $\alpha \in (0,1)$, $X\in L^{2}(\Omega, \mathcal{F},P)$ and $e_{\alpha,\varphi_2}(X)$ be the robust expectiles of $X$. Then one obtains that
 \begin{itemize}
     \item[(a)] When $\delta_2=0$, then $e_{\alpha,\varphi_2}(X)$ is the classical expectiles, i.e., 
     \begin{align*}
 e_{\alpha,\varphi_2}(X)&=\underset{m \in \mathbb{R}}{\operatorname{argmin}} \Big\{\alpha E_P\left[((X-m)^{+})^2\right]+(1-\alpha)E_P\left[((X-m)^{-})^2\right]\Big\}.
\end{align*}

\item[(b)] When $\delta_2>0$, then there exist constants $m_2^* \in \mathbb{R} $ and $\lambda^* >max\{\alpha,1-\alpha\}$, such that
$$
\left\{
\begin{aligned}
\frac{\partial g_2(X,m,\lambda,\alpha)}{\partial \lambda}\bigg|_{m=m_2^*,~\lambda=\lambda^*}=0,\\
\frac{\partial g_2(X,m,\lambda,\alpha)}{\partial m}\bigg|_{m=m_2^*,~\lambda=\lambda^*}=0,\\
\end{aligned}
\right.
$$
and $m_2^*$ is the robust expectile $e_{\alpha,\varphi_2}(X)$ and $g_2$ is defined in \eqref{eq:function-g2}.
 \end{itemize}
\end{corollary}

Indeed, Corollary \ref{cor:3.2} can be verified easily as follows. Obviously, for the given random variable $X$ and $\alpha$,  $g_2(X,m,\lambda,\alpha)$ is a binary convex function with respect to $(m,\lambda)$. To obtain the optimal solution when it reaches its extreme value, we give the partial derivatives of $g_2(X,m,\lambda,\alpha)$ with respect to $m$ and $\lambda$ respectively,
\begin{align*}
\frac{\partial g_2(X,m,\lambda,\alpha)}{\partial \lambda}=-\frac{\alpha^2}{(\lambda-\alpha)^2}E_P\left[((X-m)^{+})^2\right]-\frac{(1-\alpha)^2}{(\lambda-1+\alpha)^2}E_P\left[((X-m)^{-})^2\right]
+\delta_2,
\end{align*}
\begin{align*}
\frac{\partial g_2(X,m,\lambda,\alpha)}{\partial m}=-\frac{2\alpha\lambda}{\lambda-\alpha}E_P\left[(X-m)^{+}\right]+\frac{2(1-\alpha)\lambda}{\lambda-1+\alpha}E_P\left[(X-m)^{-}\right]
.\hspace{6em}
\end{align*}

It is noting that when $\delta_2=0$, one has that
$$\frac{\partial g_2(X,m,\lambda,\alpha)}{\partial \lambda}\leq 0 ,$$
which implies that $g_2(X,m,\lambda,\alpha)$ is decreasing about $\lambda$. Hence, $g_2(X,m,\lambda,\alpha)$ reaches the minimum value at $\lambda \rightarrow +\infty$, it follows that
\begin{align*}
\inf _{\lambda >max\{\alpha,1-\alpha\}} ~g_2(X,m,\lambda,\alpha)~&=\lim _{\lambda \to +\infty}g_2(X,m,\lambda,\alpha)\\
&=\alpha E_P\left[((X-m)^{+})^2\right]+(1-\alpha)E_P\left[((X-m)^{-})^2\right].
\end{align*}
Hence, the minimization problem reduces to
\begin{align*}
\underset{m \in \mathbb{R}}{\operatorname{argmin}} \Big\{\alpha E_P\left[((X-m)^{+})^2\right]+(1-\alpha)E_P\left[((X-m)^{-})^2\right]\Big\},
\end{align*}
which is exactly the expectiles.
\begin{remark}For $\varphi_2(x)=\infty I_{(\delta_2,+\infty)}(x)$, it means that we consider only those distributions that satisfy
$d_{c}\left(\mu_{X}, \mu\right)\leq \delta_2$.
If we take $\delta_2=0$, it implies that we only consider the distribution of uncertain future losses to be deterministic and be $\mu_{X}$. Hence, the robust expectiles lead to the classical expectiles without uncertainty distributions.
\end{remark}

Now, let's think about the case where $\delta_2>0$, one has that
$$
\lim _{\lambda \rightarrow max\{\alpha,1-\alpha\}}\frac{\partial g_2(X,m,\lambda,\alpha)}{\partial \lambda}=-\infty;\hspace{0.3em}
\lim _{\lambda \rightarrow +\infty}\frac{\partial g_2(X,m,\lambda,\alpha)}{\partial \lambda}=\delta_2>0.
$$
Thus, the claim in Corollary \ref{cor:3.2} (b) can be obtained.

\subsection{Numerical illustration}

In this subsection, we devote ourselves to numerical illustration to consider the effects of penalty parameters on robust generalized quantiles separately.

Firstly, we focus on the effect of the penalization parameter $\delta_{1}$.  Given two random variables $X$ and $Y$, suppose that the prior distributions of $X$ and $Y$ are normal distribution and  exponential distribution,  respectively.
Now, for each $\alpha \in(0,1)$,  we consider robust expectiles $e_{\alpha,\varphi_{1}}(X)$ and  $e_{\alpha,\varphi_{1}}(Y)$ for any penalization parameter $\delta_{1}>\max\{\alpha, 1-\alpha\}$.

Figure 1 depicts the $e_{\alpha,\varphi_1}(\cdot)$ with different distributions and different distribution parameters. Observing that no matter under normal distribution or exponential distribution, $e_{\alpha,\varphi_1}(X)$ and $e_{\alpha,\varphi_1}(Y)$ have the similar variation tendency. When $\alpha \in (\frac{1}{2},1)$, then $e_{\alpha,\varphi_1}(X)$ and $e_{\alpha,\varphi_1}(Y)$ decrease gradually with the increase of the penalization parameter $\delta_1$, but it is always greater than the mean value of the random losses $X$ or $Y$ respectively.  When $\alpha \in (0,\frac{1}{2})$,  $e_{\alpha,\varphi_1}(X)$ and $e_{\alpha,\varphi_1}(Y)$ are increasing gradually with the increase of penalization parameter $\delta_1$, but are always lower than the mean value of uncertainty losses $X$ or $Y$ respectively.   Moreover, for example, for random loss $X$,  it is not difficult to find that with the increase of $\delta_1$, the change of $e_{\alpha,\varphi_1}(X)$ gradually slows down and tends to be near the mean value of $X$.  This means that those distributions deviated far away from the baseline distribution have less impact on the results.   After all, the baseline distribution is the distribution that the perceptions of financial agents are closer to the true distribution, so the above results are reasonable.

\begin{figure}[H]
\centering
\subfigure[$e_{\alpha,\varphi_1}(X)$ with Normal prior distribution]{
\label{Fig.sub.1}
\includegraphics[width=7.2cm,height = 8cm]{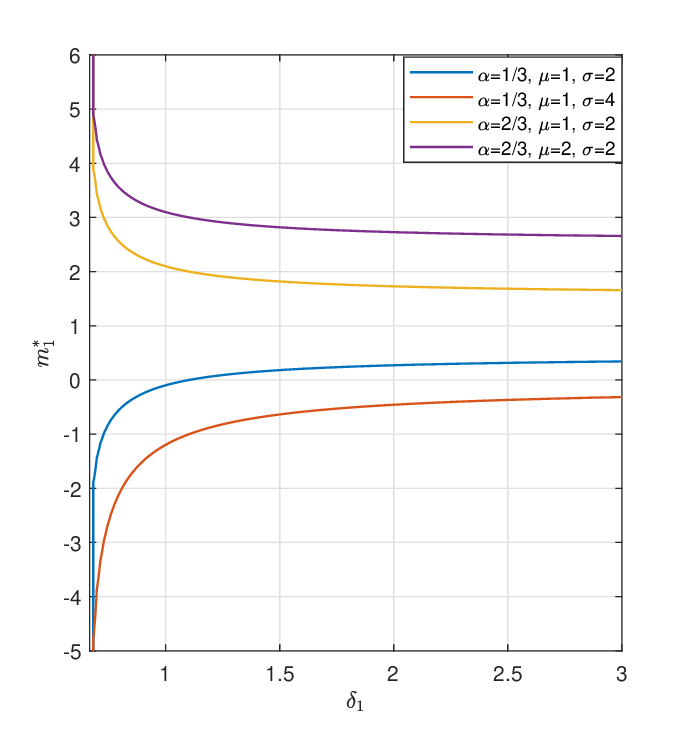}}\subfigure[$e_{\alpha,\varphi_1}(Y)$ with Exponential prior distribution]{
\label{Fig.sub.2}
\includegraphics[width=7.2cm,height =8cm]{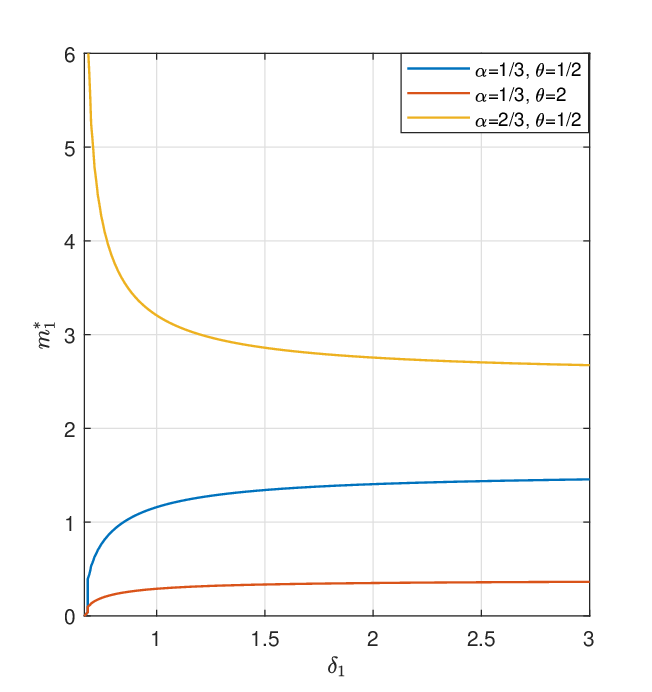}}
\caption{The effect of penalization function $\varphi_1$ on $e_{\alpha,\varphi_1}$ under different distributions}
\label{1}
\end{figure}

\begin{figure}[H]
\centering
\subfigure[$e_{\alpha,\varphi_2}(X)$ with Normal prior distribution]{
\label{Fig.sub.1}
\includegraphics[width=7.2cm,height = 8cm]{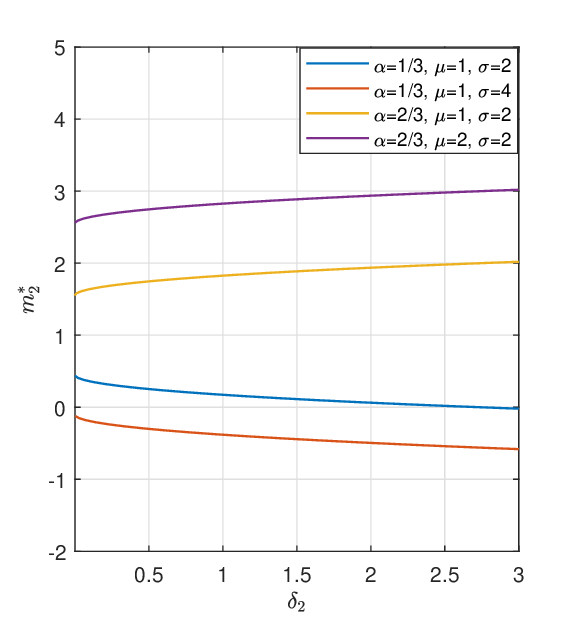}}\subfigure[$e_{\alpha,\varphi_2}(Y)$ with Exponential prior distribution]{
\label{Fig.sub.2}
\includegraphics[width=7.2cm,height =8cm]{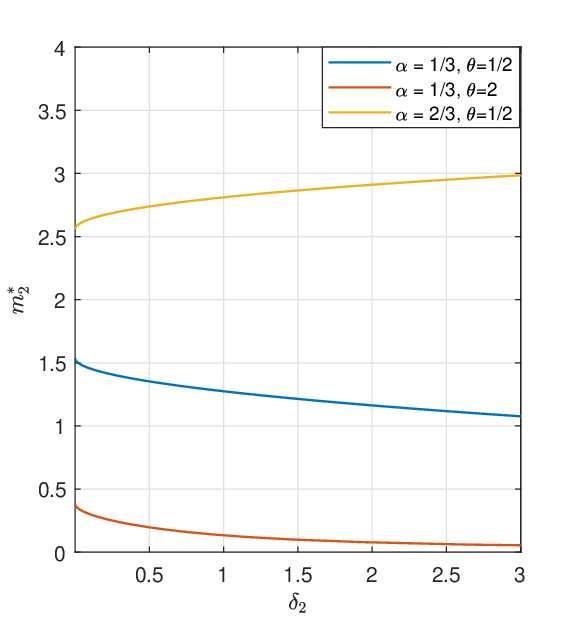}}
\caption{The effect of penalization function $\varphi_2$ on $e_{\alpha,\varphi_2}$ under different prior distributions}
\label{2}
\end{figure}

Figure 2 depicts the $e_{\alpha,\varphi_2}$ with different prior distributions and different parameters of prior distributions.  Contrary to $e_{\alpha,\varphi_1}$, when $\alpha \in (\frac{1}{2},1)$, $e_{\alpha,\varphi_2}$ is increasing with respect to $\delta_2$, while if $\alpha \in (0,\frac{1}{2})$, $e_{\alpha,\varphi_2}$ is decreasing with respect to $\delta_2$. Similarly, the degree of changes becomes slower with the increase of $\delta_2$.  This is because no matter whether we use the penalization functions $\varphi_1$ or $\varphi_2$, it tends to be that those distributions far away from the baseline distribution should not have a major impacts on our results.

According to Figure 1 and Figure 2, we have analyzed the influence of the coefficient of the penalization functions on robust expectiles under normal distribution and exponential distribution.  In fact, the change trends of robust expectiles with respect to $\delta_1$ or $\delta_2$ under the $T$ distribution are similarly consistent with that under the normal distribution. This subsection considers the differences between the classical expectiles and the robust expectiles under $T$ distribution.

\begin{figure}[H]
    \centering
    \includegraphics[width=9cm,height=8cm]{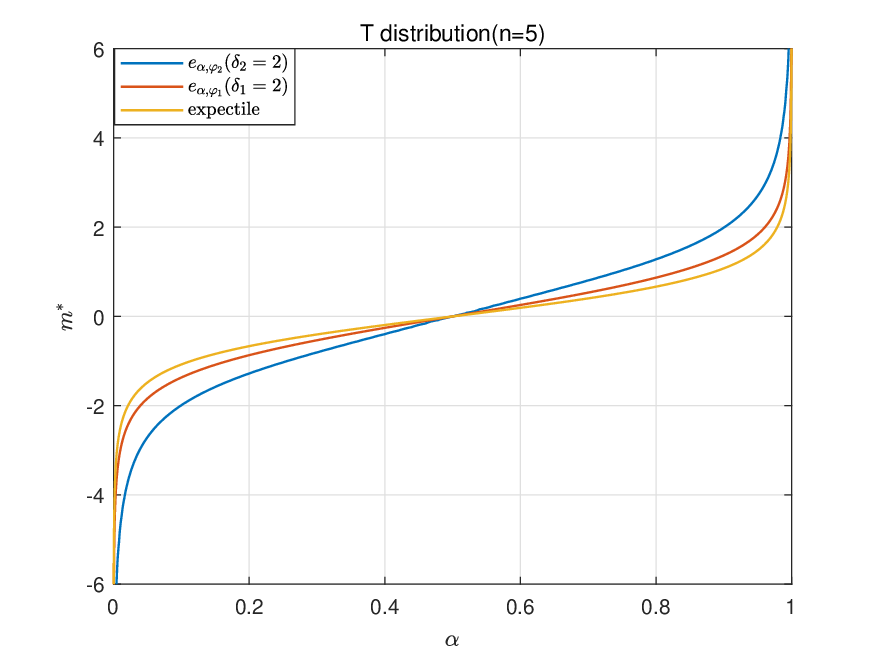}
    \caption{The comparisons between robust expectiles and expectiles under $T$ prior distribution.}
\end{figure}

Figure 3 indicates the impact on expectiles after introducing nonlinear expectation with penalty. It reveals that the trends of $e_{\alpha,\varphi_2}$, $e_{\alpha,\varphi_1}$ and expectiles with respect to $\alpha$ are the similar. For $e_{\alpha,\varphi_1}$, it is always  smaller than the expectiles when $\alpha \in (0,\frac{1}{2}]$, and it has already known that $e_{\alpha,\varphi_1}$ is increasing with respect $\delta_1$, which implies that the larger $\delta_1$ is, the closer $e_{\alpha,\varphi_1}$ is to the expectiles. While $e_{\alpha,\varphi_1}$ is always bigger than the expectiles when $\alpha \in [\frac{1}{2},1)$, and $e_{\alpha,\varphi_1}$ is decreasing with respect $\delta_1$, which can also implies that the larger $\delta_1$ is, the closer $e_{\alpha,\varphi_1}$ is to the expectiles. 

However, it is different for $e_{\alpha,\varphi_2}$.  We find that the larger $\delta_2$ is, the farther away $e_{\alpha,\varphi_2}$ is from the expectiles. This is because for $e_{\alpha,\varphi_2}$, the larger $\delta_2$ means the more distributions are considered, and the more deviations from the expectiles should be expected.

\section{The Proofs}\label{sec4}
This section provides the proofs of the previous propositions and theorems.

\subsection{The proofs for Section \ref{sec2}}

In the following, we will use the duality theorem obtained by \cite{BDT20}.

\begin{lemma}\label{lem:4-1}
(\cite{BDT20}, Theorem 2.7) The following duality theorem holds, i.e.,  for any random variable $X$,
\begin{align}\label{eq:4-1}
\mathcal{OCE}_{l, \varphi}(X)=\inf _{\lambda \geq 0}\left\{OCE_{l^{\lambda c}}(X)+\varphi^*(\lambda)\right\},
\end{align}
with
\begin{align}\label{eq:4-2}OCE_{l^{\lambda c}}(X)=\inf _{m \in \mathbb{R}}\left\{m+
\int_{\mathbb{R}}l^{\lambda c}(x-m)\mu_{X}(dx)\right\}=\inf _{m \in \mathbb{R}}\left\{m+
E_{P}[l^{\lambda c}(X-m)]\right\},
\end{align}where $l$ is a loss function, $\varphi^{*}$ is the convex conjugate of a penalization function $\varphi$ and $l^{\lambda c}$ is $\lambda c$-transform for $l$, which is defined by
$$l^{\lambda c}(x):=\sup _{y \in \mathbb{R}}\{l(y)-\lambda |x-y|^p \}, ~~x\in \mathbb{R}.$$
\end{lemma}

The following lemma gives some properties of $l^{\lambda c}(\cdot)$. 

\begin{lemma}\label{lem:4-2} Let a loss function  $l(\cdot): \mathbb{R} \to \mathbb{R}$ be convex and increasing. Then for any $\lambda \geq 0$, the $\lambda c$-transform for the loss function $l$ has the following properties:
\begin{itemize}
\item[(a)] $l^{\lambda c}(x) \geq l(x) $ for any $x \in \mathbb{R}$;
\item[(b)] $l^{\lambda c}(x)$ is increasing with respect to $x$;
\item[(c)] $l^{\lambda c}(x)$ is convex with respect to $(x,\lambda).$
\end{itemize}
\end{lemma}

\begin{proof}
By substitution of variables,  for any $\lambda \geq 0$, the $\lambda c$-transform can be expressed as
$$l^{\lambda c}(x)=\sup _{y \in \mathbb{R}}\left\{l(x-y)-\lambda |y|^p\right\}, ~~\forall x\in\mathbb{R}.$$
Hence, $(a)$  obviously holds.  Since the loss function $l$ is increasing, it implies that $(b)$ is true. For the part $(c)$, it is derived from the well known fact that an arbitrary supremum over convex functions remains convex.
\hfill$\Box$
\end{proof}


\vspace{10pt}
\noindent
\textbf{\emph{Proof of Proposition} \ref{prop:2-1}.}
The prior distribution invariance property $(a)$ is obvious derived from the definition of $\mathcal{OCE}_{l, \varphi}(\cdot)$.

For any random variable $X$, by Lemma \ref{lem:4-1},
we obtain that
\begin{align*}
\mathcal{OCE}_{l, \varphi}(X)=\inf _{\lambda \geq 0}\left\{OCE_{l^{\lambda c}}(X)+\varphi^*(\lambda)\right\},
\end{align*}
with $OCE_{l^{\lambda c}}(X)=\inf _{m \in \mathbb{R}}\left\{m+E_P[l^{\lambda c}(X-m)]\right\}.$
Therefore, for any $C\in\mathbb{R}$,
\begin{align*}
OCE_{l^{\lambda c}}(X+C)&=\inf _{m \in \mathbb{R}}\left\{m+E_P[l^{\lambda c}(X+C-m)]\right\}\\
&=\inf _{m_1 \in \mathbb{R}}\left\{m_1+C+E_P[l^{\lambda c}(X-m_1)]\right\}\\
&=OCE_{l^{\lambda c}}(X)+C,
\end{align*}
which means that $(b)$ holds.

If $X\leq Y$, $P$-a.s., since $l$ is increasing, by Lemma \ref{lem:4-2}, we know that the $\lambda c$-transform $l^{\lambda c}(\cdot)$ is increasing, then  it implies that
$$OCE_{l^{\lambda c}}(X)\leq OCE_{l^{\lambda c}}(Y).$$
Hence, $\mathcal{OCE}_{l,\varphi}(X)\leq\mathcal{OCE}_{l,\varphi}(Y),$
i.e., $(c)$ is true.

Since $l(\cdot)$ is convex,  Lemma \ref{lem:4-2} leads to $l^{\lambda c}(x)$ is convex with respect to $(x,\lambda)$.  Therefore, for any $t\in(0,1)$,  random variables $X,Y$ and $\lambda_1, \lambda_2\geq 0$,
\begin{align}\label{eq:4-3}
&\hspace{2em}OCE_{l^{(t\lambda_1+(1-t)\lambda_2) c}}(tX+(1-t)Y)\nonumber\\
&=\inf _{m_1,m_2 \in \mathbb{R}}\left\{tm_1+(1-t)m_2+E_P[l^{(t\lambda_1+(1-t)\lambda_2) c}(tX+(1-t)Y-(tm_1+(1-t)m_2))]\right\}\nonumber\\
&\leq t \inf _{m_1 \in \mathbb{R}}\left\{m_1+E_P[l^{\lambda_1 c}(X-m_1)]\right\}+(1-t)\inf _{m_2 \in \mathbb{R}}\left\{m_2+E_P[l^{\lambda_1 c}(Y-m_2)]\right\}\nonumber\\
&=tOCE_{l^{\lambda_1 c}}(X)+(1-t)OCE_{l^{\lambda_2 c}}(Y).
\end{align}

On the other hand, since $\varphi^{*}$ is the convex conjugate of the penalization function $\varphi$, then  $\varphi^*(\cdot)$ is convex.  Combing the the convexity of OCE in \eqref{eq:4-3},  it implies that
\begin{align*}
\mathcal{OCE}_{l, \varphi}(tX+(1-t)Y)&=\inf _{\lambda \geq 0}\left\{OCE_{l^{\lambda c}}(tX+(1-t)Y)+\varphi^*(\lambda)\right\}\\
&= \inf _{\lambda_1,\lambda_2 \geq 0}\left\{OCE_{l^{(t\lambda_1+(1-t)\lambda_2) c}}(tX+(1-t)Y)+\varphi^*(t\lambda_1+(1-t)\lambda_2)\right\} \\
&\leq  t\inf _{\lambda_1\geq 0}\left\{OCE_{l^{\lambda_1 c}}(X)+\varphi^*(\lambda_1)\right\}+ (1-t)\inf _{\lambda_2\geq 0}\left\{OCE_{l^{\lambda_2 c}}(Y)+\varphi^*(\lambda_2)\right\}\\
&= t\mathcal{OCE}_{l, \varphi}(X)+(1-t)\mathcal{OCE}_{l, \varphi}(Y).
\end{align*}
Thus,  $\mathcal{OCE}_{l, \varphi}(\cdot)$ is convex.  The proof is complete.   \hfill$\Box$

\vspace{10pt}

\noindent\textbf{\emph{Proof of Proposition} \ref{prop:2-4}}. $(a)$  Since $l_1$ and $l_2$ are convex, then for each $x$,  $h(x-m)$ is convex in $m$. Hence, for each $t\in(0,1)$, for all $m_1,m_2 \in \mathbb{R}$,
\begin{align*}
& \hspace{2em}\mathcal{E}_{\varphi}(h,X,tm_1+(1-t)m_2))\\
&=\sup _{\mu \in \mathcal{M}(\mathbb{R})}\left(\int_{\mathbb{R}} h(x-(tm_1+(1-t)m_2))\mu(dx)-\varphi\left(d_{c}\left(\mu_{X}, \mu\right)\right)\right)\\
&\leq \sup _{\mu \in \mathcal{M}(\mathbb{R})}\left(t \int_{\mathbb{R}} h(x-m_1)\mu(dx)+(1-t)\int_{\mathbb{R}} h(x-m_2)d\mu-\varphi\left(d_{c}\left(\mu_{X}, \mu\right)\right)\right)\\
&\leq t\sup _{\mu \in \mathcal{M}(\mathbb{R})}\left(\int_{\mathbb{R}} h(x-m_1)\mu(dx)-\varphi\left(d_{c}\left(\mu_{X}, \mu\right)\right)\right)\\
&+(1-t)\sup _{\mu \in \mathcal{M}(\mathbb{R})}\left(\int_{\mathbb{R}} h(x-m_2)\mu(dx)-\varphi\left(d_{c}\left(\mu_{X}, \mu\right)\right)\right)\\
&=t\mathcal{E}_{\varphi}(h,X,m_1)+(1-t)\mathcal{E}_{\varphi}(h,X,m_2),
\end{align*}
which implies $\mathcal{E}_{\varphi}(h,X,m)$ is convex about $m$.

On the other hand, from the definition of $\mathcal{E}_{\varphi}(h,X,m)$, it follows that
\begin{align*}
\mathcal{E}_{\varphi}(h,X,m)&=\sup _{\mu \in \mathcal{M}(\mathbb{R})}\left(\int_{\mathbb{R}} h(x-m)\mu(dx)-\varphi\left(d_{c}\left(\mu_{X}, \mu\right)\right)\right)\\
&\geq E_P[\alpha l_1((X-m)^+)+(1-\alpha)l_2((X-m)^-)].
\end{align*}
Since the monotonicity and convexity properties of the loss functions $l_{1}$ and $l_{2}$, then by the monotone convergence theorem, it derives
\begin{align*}
\lim _{m \rightarrow-\infty} E_P\left[\alpha l_1((X-m)^+)+(1-\alpha)l_2((X-m)^-)\right]=+\infty,
\end{align*} and
\begin{align*}
\lim _{m \rightarrow+\infty} E_P\left[\alpha l_1((X-m)^+)+(1-\alpha)l_2((X-m)^-)\right] =+\infty.
\end{align*}
Hence, one has that
$$
\lim _{m \rightarrow-\infty} \mathcal{E}_{\varphi}(h,X,m)=\lim _{m \rightarrow+\infty} \mathcal{E}_{\varphi}(h,X,m)=+\infty .
$$

$(b)$  Since $\mathcal{E}_{\varphi}(h,X,0)<+\infty$, then it implies  $\underset{m \in \mathbb{R}}{\operatorname{inf}}~\mathcal{E}_{\varphi}(h,X,m)<+\infty$. Due to the facts that $\mathcal{E}_{\varphi}(h,X,m)$ is convex about $m$ and $\lim _{m \rightarrow \pm\infty}\mathcal{E}_{\varphi}(h,X,m) =+\infty$, then
it is easy to find that there exists a closed interval $[m_1,m_2]$, such that
$$[m_1,m_2]=\underset{m \in \mathbb{R}}{\operatorname{argmin}}~\hspace{0.1em} \mathcal{E}_{\varphi}(h,X,m).$$\hfill$\Box$

\vspace{10pt}
\noindent\textbf{\emph{Proof of Lemma} \ref{lem:2-1}}.  Since $h$ has polynomial growth,  it means there exists a constant $C\geq 0$ such that for all $x\in \mathbb{R}$, $h(x)\leq C(1+|x|^p)$, then it implies that, for each $\lambda\geq 0$ and $x \in \mathbb{R}$,
\begin{align*}
\sup_{y\in \mathbb{R}} \{h(y)-\lambda |x-y|^p \}&\leq\sup_{y\in \mathbb{R}} \{C(1+|y|^p)-\lambda|x-y|^p \}.
\end{align*}
Obviously, there exists a $\lambda^* >C$, such that $h^{\lambda^* c}(x)<+\infty$ for all $x \in \mathbb{R}$.
\hfill$\Box$

\vspace{10pt}

\noindent\textbf{\emph{Proof of Proposition} \ref{prop:2-5}}. For any $\alpha\in(0,1)$, $m\in \mathbb{R}$,  then $$h(x-m)=\alpha (x-m)^{+}+(1-\alpha)(x-m)^{-}, ~~\forall x\in\mathbb{R}.$$
It is obvious that $h(\cdot-m)$ is bounded from below.  Based on Lemma \ref{lem:4-1}, we can obtain the $\lambda c$-transform of $h(\cdot-m)$. Directly calculations, it derives that
$$
h^{\lambda c}(x-m)= \begin{cases}\alpha(x-m)^{+}+(1-\alpha)(x-m)^{-}, & \text { if } \lambda \geq max\{\alpha,1-\alpha\}, \\ +\infty, & \text { if }\lambda < max\{\alpha,1-\alpha\},\end{cases}
$$where $\varphi^{*}$ is the convex conjugate for the penalization function $\phi$.
Then,
$$\mathcal{E}_{\varphi}(h,X,m)=\inf _{\lambda \geq max\{\alpha,1-\alpha\}}\{E_{P}\left[\alpha(X-m)^{+}+(1-\alpha)(X-m)^{-}\right]+\varphi^{*}(\lambda)\}.$$
Since $\varphi^{*}(\lambda)$ is increasing with respect to $\lambda$, then
\begin{align*}
\mathcal{E}_{\varphi}(h,X,m)=E_{P}\left[\alpha(X-m)^{+}+(1-\alpha)(X-m)^{-}\right]+\varphi^{*}(max\{\alpha,1-\alpha\}).
\end{align*}
Hence, it obvious that
 $$\underset{m \in \mathbb{R}}{\operatorname{argmin}}\left\{\mathcal{E}_{\varphi}(h,X,m)\right\}=\underset{m \in \mathbb{R}}{\operatorname{argmin}}\{\alpha E_{P}\left[(X-m)^{+}\right]+(1-\alpha) E_{P}\left[(X-m)^{-}\right]\}.$$
 By Exercise 4.4.1 in \cite{FS16}, its optimal solution $m^*\in \underset{m \in \mathbb{R}}{\operatorname{argmin}}\mathcal{E}_{\varphi}(h,X,m)$ satisfies
 $$P(X<m^*)\leq\alpha \leq P(X\leq m^*).$$
It is exactly the $VaR_\alpha(X)$ in the classical situation, which means $q_{\alpha,\varphi}(X)=VaR_{\alpha}(X)$.
\hfill$\Box$

\vspace{10pt}

At the end of this subsection, we give a sufficient condition for the 
well posedness of the $\mathcal{OCE}_{l, \varphi}$. We place it here for the purpose of keeping this complex conclusion from interfering with the presentation of the main result. 


Given a random variable $X$ and a penalization function $\varphi$. Let $l:\mathbb{R}\to \mathbb{R}$ be a convex and increasing loss function and $l^{\lambda c}$ be the $\lambda c$-transform of $l$, by Lemma \ref{lem:4-1}, we know that 
\begin{align}\label{eq:prop-4.2}
    \mathcal{OCE}_{l, \varphi}(X)=\inf _{\lambda \geq 0}\inf _{m \in \mathbb{R}}\left\{m+
E_{P}[l^{\lambda c}(X-m)]+\varphi^*(\lambda)\right\}.
\end{align}

Denote
 $$\Lambda_{l}(X,\varphi):=\left\{\lambda\geq 0 ~|~ \exists~  m\in\mathbb{R}, s.t., ~E_{P}[l^{\lambda c}(X-m)]+\varphi^*(\lambda)< +\infty\right\}. $$
If $\Lambda_{l}(X,\varphi)=\emptyset$, by the definition of robust optimized certainty equivalent and Lemma \ref{lem:4-1},   then it implies that $\mathcal{OCE}_{l,\varphi}(X)\equiv+\infty$.  To avoid trivial, suppose $\Lambda_{l}(X,\varphi)\neq\emptyset$,  and define a class of loss functions as follows:
$$
L(X,\varphi):= \left\{ \begin{array}{l} \text{loss function }\\ l:\mathbb{R}\to \mathbb{R} \end{array}: \begin{array}{l}
l  \text { is convex and increasing such that }  \\
\text{for any } \lambda \in \Lambda_{l}(X,\varphi),  l^{\lambda c}(x)\geq l^{\lambda c}(0)+x \text{ for all } x\in\mathbb{R}.
\end{array}\right\}.
$$
Then the following result shows that we can find the optional solution in the support of random variable $X$ when we choose an appropriate loss function and penalization function.

\begin{proposition}\label{prop:2-3}
Let $\varphi$ be a penalization function.  Let $X$ be random variable with compact support $suppX=[x_{min},x_{max}]$,  $(-\infty <x_{min} \leq x_{max} <+\infty)$. Then, for all $l\in L(X,\varphi)$,
\begin{align*}
\mathcal{OCE}_{l, \varphi}(X)=\min _{m \in supp X}\left\{m+\mathcal{E}_{\varphi}(l,X,m)\right\}.
\end{align*}
\end{proposition}

\begin{proof} 
By Lemma \ref{lem:4-1},
one knows that \eqref{eq:prop-4.2} holds. In view of the notation of $\Lambda_{l}(X,\varphi)$, Therefore, one can get
\begin{align}\label{eq:4-4}
\mathcal{OCE}_{l, \varphi}(X)&=\inf _{\lambda \geq 0}\left\{OCE_{l^{\lambda c}}(X)+\varphi^*(\lambda)\right\}\nonumber\\
&=\inf _{\lambda \in \Lambda_{l}(X,\varphi)}\inf_{m\in\mathbb{R}}\left\{m+
E_{P}[l^{\lambda c}(X-m)]+\varphi^*(\lambda)\right\}.
\end{align}

Step 1:  we transform the problem of finding the optimal solution for $\mathcal{OCE}_{l, \varphi}(X)$
 into the problem of locating the minimizer of $F(m)$, with $F(m):=m+E_{P}[l^{\lambda c}(X-m)]$, $m\in\mathbb{R}$,  for any $\lambda \in \Lambda_{l}(X,\varphi)$.

Indeed,  we only need to consider those $m$ that make $F(m)<+\infty$. Due to the fact that $l$ is convex, which implies $F(m)$ is convex in $m$. Since $supp X=[x_{min},x_{max}]$ with $-\infty <x_{min} \leq x_{max} <+\infty$,  then by Lebesgue's dominated convergence theorem, we can freely interchange integration with one-sided derivation. Thus, one can derive that
\begin{align*}
F^{'}_{+}(m)=1-E_{P}[(l^{\lambda c})^{'}_{+}(X-m)] \textrm{~~and~~} F^{'}_{-}(m)=1-E_{P}[(l^{\lambda c})^{'}_{-}(X-m)].
\end{align*}

Step 2: we derive the necessary conditions for $m^*$ to be an optimal solution.  By the standard results of convex analysis theory, if $m^*$ is the optimal solution of $OCE_{l^{\lambda c}}(X)$, then it should satisfy that
\begin{align*}
F^{'}_{-}(m^*)\leq 0\leq F^{'}_{+}(m^*).
\end{align*}
Or, to put it equivalently, it satisfies
\begin{align}\label{eq:mstar}
E_{P}[(l^{\lambda c})^{'}_{-}(X-m^*)]\leq 1\leq E_{P}[(l^{\lambda c})^{'}_{+}(X-m^*)].
\end{align}

Step 3: we verify that the minimum points of the problem of \eqref{eq:mstar} can be chosen in the support set of $X$. 

If $m^*<x_{min}$, then $X-m^*>X-x_{min}$. Since $l\in L(X,\varphi)$ and $l^{\lambda c}(\cdot)$ is convex, which implies $(l^{\lambda c})^{'}_{+}$ and $(l^{\lambda c})^{'}_{-}$ are nondecreasing, and we can obtain that
\begin{align*}
1\geq E_{P}[(l^{\lambda c})^{'}_{-}(X-m^*)]\geq E_{P}[(l^{\lambda c})^{'}_{+}(X-x_{min})]\geq E_{P}[(l^{\lambda c})^{'}_{-}(X-x_{min})]\geq(l^{\lambda c})^{'}_{+}(0).
\end{align*}
Similarly, if $m^*>x_{max}$, we can get
\begin{align*}
(l^{\lambda c})^{'}_{-}(0) \geq E_{P}[(l^{\lambda c})^{'}_{+}(X-x_{max})]\geq E_{P}[(l^{\lambda c})^{'}_{-}(X-x_{max})]\geq E_{P}[(l^{\lambda c})^{'}_{+}(X-m^*)]\geq 1.
\end{align*}

Since $l \in L(X,\varphi)$, then for any $\lambda \in \Lambda_{l}(X,\varphi)$, we have that $l^{\lambda c}(x)\geq l^{\lambda c}(0)+x$ for all $x\in\mathbb{R}$. Besides, from the convexity of $l^{\lambda c}(\cdot)$, one has that
\begin{align*}
(l^{\lambda c})^{'}_{+}(0)\geq 1 \geq(l^{\lambda c})^{'}_{-}(0).
\end{align*}
Combining the above three inequalities, 
we finally obtain that
\begin{align*}
E_{P}[(l^{\lambda c})^{'}_{+}(X-x_{min})]= E_{P}[(l^{\lambda c})^{'}_{-}(X-x_{min})]=1,  ~~ \textrm{~if~}~m^*<x_{min},
\end{align*}
and
\begin{align*}
E_{P}[(l^{\lambda c})^{'}_{+}(X-x_{max})]= E_{P}[(l^{\lambda c})^{'}_{-}(X-x_{max})]=1,   ~~ \textrm{~if~}~ m^*>x_{max}.
\end{align*}
Therefore, $x_{min}$ ( or $x_{max}$) is also the optimal solution of $OCE_{l^{\lambda c}}(X)$.  This fact indicates that the minimum points can be chosen in the support set of $X$. 

In the other words, for any $\lambda \in \Lambda_{l}(X,\varphi)$, we have
\begin{align*}
OCE_{l^{\lambda c}}(X)=\min _{m \in supp X}\left\{m+E_{P}[l^{\lambda c}(X-m)]\right\}.
\end{align*}
Then, according to \eqref{eq:4-4}, one has that
\begin{align*}
\mathcal{OCE}_{l, \varphi}(X)&=\inf _{\lambda \in \Lambda_{l}(X,\varphi)}\left\{\min _{m \in supp X}\left\{m+E_{P}[l^{\lambda c}(X-m)]\right\}+\varphi^*(\lambda)\right\}\\
&=\inf _{\lambda \geq 0}\left\{\min _{m \in supp X}\left\{m+E_{P}[l^{\lambda c}(X-m)]\right\}+\varphi^*(\lambda)\right\}\\
&=\min _{m \in supp X}\left\{m+\inf _{\lambda \geq 0}\left\{\varphi^*(\lambda)+E_{P}[l^{\lambda c}(X-m)]\right\}\right\}\\
&=\min _{m \in supp X}\left\{m+\mathcal{E}_{\varphi}(l,X,m)\right\}.
\end{align*}The proof is complete. \hfill$\Box$
\end{proof}

\subsection{The proofs for Section \ref{sec3}}

\noindent \textbf{\emph{Proof of Proposition} \ref{prop:3-1}}. For each $\alpha\in(0,1)$ and $m\in\mathbb{R}$, denote
$$h(x-m)=\alpha((x-m)^{+})^2+(1-\alpha))((x-m)^{-})^2,  \forall x\in\mathbb{R}.$$
Then, the $\lambda c$-transform of loss function $h(\cdot-m)$ can be calculated as follows and it can be divided into three cases.

Case (i):  When $\alpha \in (0,\frac{1}{2})$, then
$$
h^{\lambda c}(x-m)= \begin{cases}+\infty, & \text { if } \lambda< 1-\alpha, \\
\\
\frac{\alpha (1-\alpha) }{1-2\alpha}((x-m)^{+})^2+\infty I_{(-\infty, m)}(x), & \text { if } \lambda= 1-\alpha,\\
\\
\frac{\alpha\lambda}{\lambda-\alpha}((x-m)^{+})^2+\frac{(1-\alpha)\lambda}{\lambda-(1-\alpha)}((x-m)^{-})^2, & \text { if }\lambda > 1-\alpha. \end{cases}
$$

Case (ii): When $\alpha =\frac{1}{2}$, then
$$
h^{\lambda c}(x-m)= \begin{cases}+\infty, & \text { if } \lambda \leq \frac{1}{2}, \\
\\
\frac{\alpha\lambda}{\lambda-\alpha}((x-m)^{+})^2+\frac{(1-\alpha)\lambda}{\lambda-(1-\alpha)}((x-m)^{-})^2, & \text { if }\lambda > \frac{1}{2}. \end{cases}
$$

Case (iii): When $\alpha \in (\frac{1}{2},1)$, then
$$
h^{\lambda c}(x-m)= \begin{cases}+\infty, & \text { if } \lambda< \alpha, \\
\\
\infty I_{(m,+\infty)}(x)+\frac{(1-\alpha)\alpha}{2\alpha-1}((x-m)^{-})^2, & \text { if } \lambda= \alpha,\\
\\
\frac{\alpha\lambda}{\lambda-\alpha}((x-m)^{+})^2+\frac{(1-\alpha)\lambda}{\lambda-(1-\alpha)}((x-m)^{-})^2, & \text { if }\lambda > \alpha. \end{cases}
$$

Since $\varphi_1(x)=\delta_1 x$ , one has that $\varphi_1^{*}(\lambda)=\infty I_{(\delta_1,+\infty)}(\lambda)$. Therefore, by Lemma \ref{lem:4-1}, it follows that
\begin{align*}
\mathcal{E}_{\varphi_1}(h,X,m)&=\sup _{\mu \in \mathcal{M}(\mathbb{R})}\left(\int_{\mathbb{R}} h(x-m)\mu(dx)-\delta_1d_{c}\left(\mu_{X}, \mu\right)\right)\\
&=\inf _{\lambda \geq 0 }\left\{E_P[h^{\lambda c}(X-m)]+\infty I_{(\delta_1,+\infty)}(\lambda) \right\} \\
&=\inf _{0 \leq \lambda \leq \delta_1 }~E_P[h^{\lambda c}(X-m)]\\
&=\inf _{max\{\alpha,1-\alpha\} < \lambda \leq \delta_1 }g_1(X,m,\lambda,\alpha),
\end{align*}
where
$$g_1(X,m,\lambda,\alpha)=\frac{\alpha\lambda}{\lambda-\alpha}E_P\left[((X-m)^{+})^2\right]+\frac{(1-\alpha)\lambda}{\lambda-(1-\alpha)}E_P\left[((X-m)^{-})^2\right].$$
It is clear that $g_1(X,m,\lambda,\alpha)$ is decreasing with respect to $\lambda$. Thus, one has that
$$\inf _{max\{\alpha,1-\alpha\} < \lambda \leq \delta_1 }g_1(X,m,\lambda,\alpha)=g_1(X,m,\delta_{1},\alpha).$$
\hfill$\Box$

\noindent \textbf{\emph{Proof of Theorem} \ref{th:3-1}.}
For any $\alpha\in (\frac{1}{2}, 1)$ and $\delta_{1}>\alpha$, we verify that $e_{\alpha,\varphi_1}(\cdot)$ satisfies the axioms of the coherent risk measures.

\emph{(i)} Translation invariance.    For any constants $C$ and $m$ in $\mathbb{R}$, since
 $$g_1(X+C,m,\delta_{1}, \alpha)=g_1(X,m-C, \delta_{1},\alpha).$$
Hence, by the definition of robust expectile, then it implies that
$$e_{\alpha,\varphi_1}(X+C)=e_{\alpha,\varphi_1}(X)+C.$$

\emph{(ii)} Monotonicity.
Since $X\in L^{2}(\Omega, \mathcal{F},P)$, then it can be verified that $g_1(X,m,\delta_{1},\alpha)$ is differentiable with respect to $m$, and
\begin{align*}
g_1'(X,m,\delta_{1},\alpha)=-\frac{2\alpha\delta_1}{\delta_1-\alpha}E_P\left[(X-m)^{+}\right]+\frac{2(1-\alpha)\delta_1}{\delta_1-(1-\alpha)}E_P\left[(X-m)^{-}\right], ~\forall m\in\mathbb{R}.
\end{align*}
Since $e_{\alpha,\varphi_1}(X)= \underset{m \in \mathbb{R}}{\operatorname{argmin}} ~g_1(X,m,\delta_{1},\alpha)$ and  $g_1(X,\cdot,\delta_{1},\alpha)$ is convex, then its optimal value for $m$ should satisfy
$$g_1'(X,m,\delta_{1},\alpha)=0.$$
Since $\delta_1>max\{\alpha,1-\alpha\}$, one has that
$$-\frac{2\alpha\delta_1}{\delta_1-\alpha}<0 ~~\textrm{and}~~\frac{2(1-\alpha)\delta_1}{\delta_1-(1-\alpha)}>0.$$
Then, for the given $\delta_{1},m$ and $\alpha$, $g_1'(X,m,\delta_{1},\alpha)$ is decreasing with respect to $X$, i.e.,
if  $X\leq Y$, $P$-a.s., then $g_1'(X,m,\delta_{1},\alpha)\geq g_1'(Y,m,\delta_{1},\alpha)$.

Therefore, when $X\leq Y$, $P$-a.s., then
$$g_1'(Y,e_{\alpha,\varphi_1}(X),\delta_{1},\alpha)\leq g_1'(X,e_{\alpha,\varphi_1}(X),\delta_{1},\alpha)=0 .$$
Since $g_1'(Y,m,\delta_{1},\alpha)$ is increasing with respect to $m$, and $g_1'(Y,e_{\alpha,\varphi_1}(Y),\delta_{1},\alpha)=0$, then it implies that
$$e_{\alpha,\varphi_1}(X)\leq e_{\alpha,\varphi_1}(Y).$$

\emph{(iii)} Convexity.   Since $\alpha \in (\frac{1}{2},1)$ and $\delta_1>max\{\alpha,1-\alpha\}$, then it can derive that
$$\frac{2\alpha\delta_1}{\delta_1-\alpha} \geq \frac{2(1-\alpha)\delta_1}{\delta_1-(1-\alpha)},$$
which implies that $g_1'(X,m, \delta_{1},\alpha)$ is concave with respect to $(X,m)$.  Recall that
$$g_1'(X,e_{\alpha,\varphi_1}(X), \delta_{1},\alpha)=0 \textrm{ and } g_1'(Y,e_{\alpha,\varphi_1}(Y), \delta_{1},\alpha)=0.$$
For any $t\in (0,1)$, it then implies
\begin{align*}
& g_1'\left(tX+(1-t)Y,  t e_{\alpha,\varphi_1}(X)+(1-t) e_{\alpha,\varphi_1}(Y), \delta_{1},\alpha\right)\\
\geq &~ t   g_1'\left(X,e_{\alpha,\varphi_1}(X), \delta_{1},\alpha\right)+(1-t) g_1'\left(Y,e_{\alpha,\varphi_1}(Y), \delta_{1},\alpha\right)\\
=& 0.
\end{align*}
Since $g_1'(X,m, \delta_{1},\alpha)$ is increasing with respect to $m$, and
$$g_1'(tX+(1-t)Y,e_{\alpha,\varphi_1}(tX+(1-t)Y), \delta_{1},\alpha)=0.$$
Therefore, for the given random variables $X$ and $Y$, for each $t\in (0,1)$, $\alpha \in (\frac{1}{2},1)$, and $\delta_1>max\{\alpha,1-\alpha\}$, then

$$e_{\alpha,\varphi_1}(tX+(1-t)Y)\leq t e_{\alpha,\varphi_1}(X)+(1-t) e_{\alpha,\varphi_1}(Y).$$

\emph{(iv)} Positive homogeneity.  When $t=0$, it is obvious that $e_{\alpha,\varphi_1}(0)=0$.   For any $t>0$, one has that
$$g_1'(tX,m, \delta_{1},\alpha)=tg_1'(X,\frac{m}{t}, \delta_{1},\alpha).$$
Hence, for any $t\geq 0$,  $e_{\alpha,\varphi_1}(tX)=t \hspace{0.2em} e_{\alpha,\varphi_1}(X)$.  \hfill$\Box$

\vspace{10pt}

\noindent \textbf{\emph{Proof of Theorem} \ref{th:3-2}}. Suppose $\alpha \in (\frac{1}{2},1)$. By Theorem \ref{th:3-1},  $e_{\alpha,\varphi_1}(X)$ is also a convex risk measure.  Based on Proposition 4.113 and Theorem 4.115 in \cite{FS16},   $e_{\alpha,\varphi_1}(X)$ can be represented as
$$e_{\alpha,\varphi_1}(X)=\underset{Q \in \mathcal{M}(P_X)}{\operatorname{max}}\left\{E_{Q}[X]- \inf _{t>0}\frac{1}{t}E_{P}\left[\psi^{*}(t\frac{\mathrm{d}Q}{\mathrm{d}P})\right]  \right\},$$
where
$$
\mathcal{M}(P)=\left\{Q:
Q \textrm{ is  absolutely  continuous  with respect to } P \right\},
$$and  $\psi^{*}$ is the dual conjugate function of $\psi$ with
$$
\psi^{*}(x)=
\begin{cases} 0  & \text {  } \frac{2(1-\alpha)\delta_1}{\delta_1-(1-\alpha)}\leq x \leq \frac{2\alpha\delta_1}{\delta_1-\alpha}, \\
+\infty & \text {  } else. \\
\end{cases}
$$
Hence, it derives that
$$\inf _{t>0}\frac{1}{t}E_{P}[\psi^{*}(t\frac{\mathrm{d}Q}{\mathrm{d}P})]=\begin{cases} 0  & \exists t_{0}>0, s.t., \frac{2(1-\alpha)\delta_1}{\delta_1-(1-\alpha)}\leq t_0\frac{\mathrm{d}Q}{\mathrm{d}P} \leq \frac{2\alpha\delta_1}{\delta_1-\alpha},\\
+\infty & \text {  } else. \\
\end{cases}
$$
Thus,  $$e_{\alpha,\varphi_1}(X)=\underset{\mu \in \mathcal{M}_1(P)}{\operatorname{max}} E_{Q}[X].$$

On the other hand, since $e_{\alpha,\varphi_1}(X)=-e_{1-\alpha,\varphi_1}(-X)$, it is easy to obtain the dual representation of $e_{\alpha,\varphi_1}(X)$ when $\alpha\in(0,\frac{1}{2})$.  \hfill$\Box$

\vspace{10pt}

\noindent \textbf{\emph{Proof of Proposition} \ref{prop:3-2}}.  Since $\varphi_2(x)=\infty I_{(\delta_2,+\infty)}(x)$, then the convex conjugate function of $\varphi_{2}$ is
 $$\varphi_2^{*}(y)=\delta_2 y^{+},~~y\in\mathbb{R}.$$
By Lemma \ref{lem:4-1}, it implies that
\begin{align*}
\mathcal{E}_2(h,X,m)&= \sup _{\mu \in \mathcal{M}(\mathbb{R}),
d_{c}(\mu_{X}, \mu)\leq \delta_{2}}
\int_{\mathbb{R}} h(x-m)\mu(dx) \\
&=\inf _{\lambda \geq 0 }\left\{E_P[h^{\lambda c}(X-m)]+\delta_2\lambda \right\} \\
&=\inf _{\lambda >max\{\alpha,1-\alpha\} }g_2(X,m,\lambda),
\end{align*}where $h^{\lambda c}$ can be found in Proposition  \ref{prop:3-1}. We complete the proof.
\hfill$\Box$

\section{The Conclusions}\label{sec5}

Inspired by \cite{BDT20}, the paper analyzes the relevant properties of the robust optimized certainty equivalents as risk measures for loss positions with distribution uncertainty.   Based on the robust optimized certainty equivalents, we propose the robust generalized quantiles, which is a natural generalization for the quantiles.   Furthermore, we focus on two kinds of specific robust expectiles corresponding to two penalization functions $\varphi_{1}$ and $\varphi_{2}$.
The robust expectiles with  $\varphi_{1}$ are proved to be coherent risk measures, and the dual representation theorems are established.  The results are a development and complement to \cite{BKMR14} and \cite{BDT20}.  Besides, we also study the influences of penalization functions on the robust expectiles and compare them with expectiles for some specific prior distributions by numerical simulations.



\end{document}